\documentclass[12pt]{article}

\usepackage{amsmath}
\usepackage{amssymb}
\usepackage{amsthm}
\usepackage{latexsym}
\usepackage{cite}
\usepackage{psfrag}
\usepackage{epsfig}
\usepackage{graphicx}
\usepackage{color}
\usepackage{amsfonts}
\usepackage{mathrsfs}
\usepackage{url}
\usepackage{hyperref}

\numberwithin{equation}{section}

\newtheorem{theorem}{Theorem}[section]
\newtheorem{corollary}[theorem]{Corollary}

\newtheorem{lemma}[theorem]{Lemma}
\theoremstyle{definition}
\newtheorem{definition}[theorem]{Definition}

\newtheorem{notation}[theorem]{Notation}

\def\A{\mathcal{A}}
\def\CC{\mathcal{C}}

\def\V{\mathcal{V}}

\def\F{\mathbb{F}}

\def\c{\mathbf{c}}

\def\v{\mathbf{v}}
\def\x{\mathbf{x}}

\begin{document}
\title{\textbf{On integral weight spectra of the MDS codes cosets of weight 1, 2, and 3}
\author{Alexander A. Davydov\footnote{The research of A.A.~Davydov was done at IITP RAS and supported by the Russian Government (Contract No 14.W03.31.0019).} \\
{\footnotesize Institute for Information Transmission Problems
(Kharkevich
institute), Russian Academy of Sciences}\\
{\footnotesize Bol'shoi Karetnyi per. 19, Moscow,
127051, Russian Federation. E-mail: adav@iitp.ru}
\and Stefano Marcugini\footnote{The research of  S. Marcugini and F.~Pambianco was
 supported in part by the Italian
National Group for Algebraic and Geometric Structures and their Applications (GNSAGA - INDAM) (No U-UFMBAZ-2017-000532, 28.11.2017) and by
University of Perugia (Project: Curve, codici e configurazioni di punti, Base Research
Fund 2018,  No 46459, 15.06.2018).}  \,and Fernanda Pambianco$^\dag$  \\
{\footnotesize Department of Mathematics and Computer Science, University of Perugia}\\
{\footnotesize Via Vanvitelli~1, Perugia, 06123, Italy. E-mail:
\{stefano.marcugini,fernanda.pambianco\}@unipg.it}}
}
\date{}
\maketitle
\textbf{Abstract.}
The weight of a coset of a code is the smallest Hamming weight of any vector in the coset. For a linear code of length $n$, we call \emph{integral weight spectrum} the overall numbers of weight $w$ vectors, $0\le w\le n$, in all the cosets of a fixed weight. For maximum distance separable (MDS) codes, we obtained new convenient formulas of integral weight spectra of cosets of weight 1 and 2. Also, we give the spectra for the weight 3 cosets of MDS codes with minimum distance $5$ and covering radius $3$.


Mathematics Subject Classication (2010). 94B05, 51E21, 51E22

\section{Introduction}\label{sec_Intro}
Let $\F_{q}$ be the Galois field with $q$ elements, $\F_{q}^*=\F_{q}\setminus\{0\}$. Let $\F_{q}^{n}$ be
the space of $n$-dimensional vectors over ${\mathbb{F}}_{q}$.  We denote by  $[n,k,d]_{q}R$ an $\F_q$-linear code of length $n$, dimension $k$, minimum distance~$d$, and covering radius $R$. If $d=n-k+1$,   it is a maximum distance separable (MDS) code.
 For an introduction to coding theory see \cite{Blahut,MWS,HufPless,Roth}.

A \emph{coset} of a code is a translation of the code. A coset $\V$ of an $[n,k,d]_{q}R$ code $\CC$ can be represented as
\begin{align}\label{eq1_coset}
  \V=\{\x\in\F_q^n\,|\,\x=\c+\v,\c\in \CC\}\subset\F_q^n
\end{align}
 where $\v\in \V$ is a vector  fixed for the given representation; see  \cite{Blahut,HufPless,MWS,HandbookCodes,Roth} and the references therein.

 The weight distribution of code cosets is an important combinatorial property of a code. In particular, the distribution serves to estimate decoding results.  There are many papers connected with distinct aspects of the weight distribution of cosets for codes over distinct fields and rings, see e.g. \cite{AsmMat,Blahut,BonnDuursma,CharpHelZin,CheungIEEE1989,CheungIEEE1992,Delsarte4Fundam,DelsarteLeven,Helleseth,JurrPellik,KaipaIEEE2017,KasLin,Schatz1980,MW1963,ZDK_DHRIEEE2019}, \cite[Sect.\,6.3]{DelsarteBook}, \cite[Sect.\,7]{HufPless}, \cite[Sections 5.5, 6.6, 6.9]{MWS}, \cite[Sect.\,10]{HandbookCodes} and the references therein.

For a linear code of length $n$, we call \emph{integral weight spectrum} the overall numbers of weight $w$ vectors, $0\le w\le n$, in all the cosets of a fixed weight.

\emph{In this work,} for MDS codes, using and developing the results of \cite{CheungIEEE1989}, we obtain new convenient formulas of integral weight spectra of cosets of weight 1 and 2.
The obtained formulas for weight 1 and 2 cosets, seem to be  simple and  expressive.

This paper is organized as follows. Section \ref{sec_prelimin} contains preliminaries. In Section \ref{sec_wd1}, we consider the integral weight spectrum of the weight 1 cosets of MDS codes with minimum distance $d\ge3$. In Section \ref{sec_wd2}, we obtain the integral weight spectrum of the weight 2 cosets of MDS codes with minimum distance $d\ge5$. In Section \ref{sec_wd3}, we give the spectra for the weight 3 cosets of MDS codes with minimum distance $5$ and covering radius $3$.

\section{Preliminaries}\label{sec_prelimin}

\subsection{Cosets of a linear code}\label{subsec_cosets}
We give a few known definitions and properties connected with cosets of linear codes, see e.g. \cite{Blahut,HufPless,MWS,HandbookCodes,Roth} and the references therein.

We consider a coset $\V$ of an $[n,k,d]_{q}R$ code $\CC$ in the form \eqref{eq1_coset}. We have $\#\V=\#\CC=q^k$. One can take as $\v$ any vector of $\V$. So, there are $\#\V=q^k$ formally distinct representations of the form~\eqref{eq1_coset}; all they give the same coset $\V$. If $\v\in \CC$, we have $\V=\CC$. The distinct cosets of~$\CC$ partition $\F_{q}^{n}$ into $q^{n-k}$ sets of size $q^k$.

We remind that the \emph{Hamming weight} of the vector $\x\in \F_q^n$ is the number of nonzero entries in $\x$.

\begin{notation}\label{notation_coset}
For an $[n,k,d]_{q}R$ code $\CC$ and its coset $\V$ of the form \eqref{eq1_coset}, the following notation is used:
\begin{align*}
&t=\left\lfloor\frac{d-1}{2}\right\rfloor&&\text{the number of correctable errors};\displaybreak[3]\\
&A_w(\CC)&&\text{the number of weight $w$ codewords of the code $\CC$;}\displaybreak[3]\\
&A_w(\V)&&\text{the number of weight $w$ vectors in the coset }\V;\displaybreak[3]\\
&\text{the weight of a coset}&&\text{the smallest Hamming weight of any vector in the coset;}\displaybreak[3]\\
&\V^{(W)}&&\text{a coset of weight }W;~~A_w(\V^{(W)})=0\text{ if }w<W;\displaybreak[3]\\
&\text{a coset leader}&&\text{a vector in the coset of the smallest Hamming weight;}\displaybreak[3]\\
&\A_w^{\Sigma}(\V^{(W)})&&\text{the overall number of weight $w$ vectors in all cosets of weight }W;\displaybreak[3]\\
&\A_w^{\Sigma}(\V^{\le W})&&\text{the overall number of weight $w$ vectors in all cosets of weight }\le W.
\end{align*}
\end{notation}

In cosets of weight $> t$, a  vector of the minimal weight is not necessarily unique. Cosets of weight $\leq t$ have a unique leader.

 The code $\CC$  is the coset of weight zero.  The leader of $\CC$ is the zero vector of $ \F_q^{n}$.
\begin{definition} Let $\CC$ be an $[n,k,d]_{q}R$ code and let $\V^{(W)}$ be its coset of weight $W$. Let $\A_w^{\Sigma}(\V^{(W)})$ be the overall number of weight $w$ vectors in all cosets of weight $W$.
For a fixed~$W$, we call the set $\{\A_w^{\Sigma}(\V^{(W)})|w=0,1,\ldots,n\}$
\emph{integral weight spectrum} of the code cosets of weight $W$.
\end{definition}

Distinct representations of the integral weight spectra $\A_w^{\Sigma}(\V^{(W)})$ and values of $\A_w^{\Sigma}(\V^{\le W})$ are considered in the literature, see e.g. \cite[Th.\,14.2.2]{Blahut}, \cite{CheungIEEE1989,CheungIEEE1992}, \cite[Lem.\,2.14]{MW1963}, \cite[Th.\,6.22]{MWS}. For instance, in \cite[Eqs.\,(11)--(13)]{CheungIEEE1989}, for an MDS code correcting $t$-fold errors, the value $D_u$ gives  $\A_{u}^{\Sigma}(\V^{\le t})$.

\subsection{Some useful relations}
For $w\ge d$, the weight distribution $A_{w}(\CC)$ of an $[n,k,d=n-k+1]_q$ MDS code $\CC$ has the following form, see e.g. \cite[Th.\,7.4.1]{HufPless}, \cite[Th.\,11.3.6]{MWS}:
\begin{align}\label{eq2_wd_MDS}
 A_{w}(\CC)=\binom{n}{w}\sum_{j=0}^{w-d}(-1)^j\binom{w}{j}(q^{w-d+1-j}-1).
\end{align}

In $\F_q^n$, the volume of a sphere of radius $t$ is
\begin{align}\label{eq2_sphere}
 V_n(t)=\sum_{i=0}^t(q-1)^i\binom{n}{i}.
\end{align}

The following combinatorial identities are well known, see e.g. \cite[Sect.\,1, Eqs.\,(I),(IV), Problem 9(a)]{Riordan}:
\begin{align}
 &\binom{n}{k}=\binom{n-1}{k}+\binom{n-1}{k-1}.\label{eq2_Riordan_ident0}\displaybreak[3]\\
 &\binom{n}{m}\binom{m}{p}=\binom{n}{p}\binom{n-p}{m-p}=\binom{n}{m-p}\binom{n-m+p}{p}.\label{eq2_Riordan_ident}\displaybreak[3]\\
 &\sum_{k=0}^m(-1)^k\binom{n}{k}=(-1)^m\binom{n-1}{m}.\label{eq2_Riordan_ident2}
\end{align}

In \cite[Eqs.\,(11)--(13)]{CheungIEEE1989}, for an $[n,k,d\ge 2t+1]_q$ MDS code correcting $t$-fold errors, the following relations for $\A_u^{\Sigma}(\V^{\le t})$ denoted by $D_u$ are given:
\begin{align}\label{eq2_Cheung}
& \A_u^{\Sigma}(\V^{\le t})=D_u=\binom{n}{u}\sum_{j=0}^{u-d+t}(-1)^jN_j,~d-t\le u\le n,\displaybreak[3] \\
&\text{where}\displaybreak[3]\notag \\
&N_j=\binom{u}{j}\left[q^{u-d+1-j}V_n(t)-\sum_{i=0}^t\binom{u-j}{i}(q-1)^i\right]~\text{ if }~0\le j\le u-d,\label{eq2_Cheung_2}\displaybreak[3] \\
&N_j=\binom{u}{j}\left[\,\sum_{w=d-u+j}^t\binom{n-u+j}{w}\sum_{i=0}^{w-d+u-j}(-1)^i\binom{w}{i}(q^{w-d+u-j-i+1}-1)\right.\label{eq2_Cheung_3}\displaybreak[3]\\
&\left.\times\sum_{s=w}^t\binom{u-j}{s-w}(q-1)^{s-w}\right]~\text{ if }~u-d+1\le j \le u-d+t.\notag
\end{align}

\section{The integral weight spectrum of the weight 1 cosets of MDS codes with minimum distance $d\ge3$}\label{sec_wd1}

In Sections \ref{sec_wd1}--\ref{sec_wd3}, we represent the  values $\A_w^{\Sigma}(\V^{(W)})$ in distinct forms that can be convenient in distinct utilizations, e.g. for estimates of the decoder error probability, see \cite{CheungIEEE1989,CheungIEEE1992} and the references therein.

 We use (with some transformations) the results of \cite[Eqs.\,(11)--(13)]{CheungIEEE1989} where, for an MDS code correcting $t$-fold errors, the value $D_u$ gives the overall number $\A_{u}^{\Sigma}(\V^{\le t})$ of weight $u$ vectors in all cosets of weight $\le t$.
We cite \cite[Eqs.\,(11)--(13)]{CheungIEEE1989} by formulas \eqref{eq2_Cheung}--\eqref{eq2_Cheung_3}, respectively.

In the rest of the paper we put that a sum $\sum_{i=0}^A\ldots$ is equal to zero if $A<0$.
\begin{lemma}
\label{lem4_AwSigma<=1MDS}
\emph{\cite[Eqs.\,(11)--(13)]{CheungIEEE1989}}
Let $d-1\le w\le n$. For an $[n,k,d=n-k+1]_q$ MDS code $\CC$ of minimum distance $d\ge3$,  the overall number $\A_{w}^{\Sigma}(\V^{\le1})$ of weight $w$ vectors in all cosets of weight $\le1$ is as follows:
\begin{align}\label{eq4_AwSigma<=1MDS}
&\A_w^{\Sigma}(\V^{\le1})=\binom{n}{w}\left[\sum_{j=0}^{w-d}(-1)^j\binom{w}{j}
\left[q^{w-d+1-j}(1+n(q-1))-1-(w-j)(q-1)\right]\right.\displaybreak[3]\\
&\left.-(-1)^{w-d}\binom{w}{d-1}(n-d+1)(q-1)\right].\notag
 \end{align}
\end{lemma}
\begin{proof}
  In the relations for $D_u$ of \cite{CheungIEEE1989}  cited by \eqref{eq2_Cheung}--\eqref{eq2_Cheung_3}, we put $t=1$ and then use \eqref{eq2_sphere}. In~\eqref{eq2_Cheung_3}, we have $j=u-d+1$ whence $w=1$ in all terms. Finally, we change $u$ by $w$ to save the notations of this paper.
\end{proof}

\begin{lemma}\label{lem4_Riord}
The following holds:
\begin{align}\label{eq4_combin_equal}
\sum_{j=0}^{m}(-1)^j\binom{w}{j}\binom{w-j}{v}=(-1)^{m}\binom{w}{v}\binom{w-v-1}{m}.
\end{align}
\end{lemma}
\begin{proof}
 In \eqref{eq2_Riordan_ident}, we put $n=w$, $p=j$, $m-p=v$, and obtain
   \begin{align*}
 \sum_{j=0}^{m}(-1)^j\binom{w}{j}\binom{w-j}{v}=\binom{w}{v}\sum_{j=0}^{m}(-1)^j\binom{w-v}{j}.
   \end{align*}
Now we use \eqref{eq2_Riordan_ident2}.
\end{proof}
\begin{lemma}\label{lem4_-1}
Let $d-1\le w\le n$. The following holds:
\begin{align*}
 &\sum_{j=0}^{w+1-d}(-1)^j\binom{w}{j}q^{w+1-d-j}=\sum_{j=0}^{w-d}(-1)^j\binom{w}{j}\left(q^{w+1-d-j}-1\right)-(-1)^{w-d}\binom{w-1}{d-2}.
\end{align*}
\end{lemma}
\begin{proof} We write the left sum of the assertion as
\begin{align*}
\sum_{j=0}^{w-d}(-1)^j\binom{w}{j}\left(q^{w+1-d-j}-1+1\right)-(-1)^{w-d}\binom{w}{d-1}.
\end{align*}
By \eqref{eq2_Riordan_ident2},
\begin{align*}
\sum_{j=0}^{w-d}(-1)^j\binom{w}{j} =(-1)^{w-d}\binom{w-1}{d-1}.
\end{align*}
 Finally, we apply \eqref{eq2_Riordan_ident0}.
\end{proof}
For an $[n,k,d]_q$ code $\CC$, we denote
\begin{align}
\Omega_w^{(j)}(\CC)=(-1)^{w-d}\binom{n-j}{w-j}\binom{w-j-1}{d-j-2}.\label{eq4_Omega}
\end{align}
Also, we denote
\begin{align}\label{eq4_Phi}
&\Phi_w^{(j)}=(-1)^{w-5}\left[\binom{q+1}{w}\binom{w-1}{3}-\binom{q+1-j}{w-j}\binom{w-1-j}{3-j}\right].
\end{align}
\begin{theorem}\label{th4_Aw1sum}
\emph{\textbf{(integral weight spectrum 1)}}

Let $d-1\le w\le n$.
Let $\CC$ be an $[n,k,d=n-k+1]_q$ MDS code of minimum distance $d\ge3$.

\textbf{\emph{(i)}} For the code $\CC$, the overall number $\A_{w}^{\Sigma}(\V^{(1)})$ of weight $w$ vectors in all weight $1$ cosets is as follows:
\begin{align}\label{eq4_Aw1sum}
&\A_{w}^{\Sigma}(\V^{(1)})=\binom{n}{w}(q-1)\left[n\sum_{j=0}^{w+1-d}(-1)^j\binom{w}{j}q^{w+1-d-j}+(-1)^{w-d}w\binom{w-2}{d-3}\right]\displaybreak[3]\\
&=n(q-1)\left[\binom{n}{w}\sum_{j=0}^{w+1-d}(-1)^j\binom{w}{j}q^{w+1-d-j}+\Omega_w^{(1)}(\CC)\right]\label{eq4_Aw1sum_2}\displaybreak[3]\\
&=n(q-1)\left[\binom{n}{w}\sum_{j=0}^{w-d}(-1)^j\binom{w}{j}\left(q^{w+1-d-j}-1\right)-\Omega_w^{(0)}(\CC)+\Omega_w^{(1)}(\CC)\right]\label{eq4_Aw1sum_3}\displaybreak[3]\\
&=n(q-1)\left[A_w(\CC)-\Omega_w^{(0)}(\CC)+\Omega_w^{(1)}(\CC)\right]\label{eq4_Aw1sum_4}\displaybreak[3]\\
&=n(q-1)\left[A_w(\CC)-(-1)^{w-d}\left(\binom{n}{w}\binom{w-1}{d-2}-\binom{n-1}{w-1}\binom{w-2}{d-3}\right)\right].\label{eq4_Aw1sum_5}
\end{align}

\textbf{\emph{(ii)}} Let the code $\CC$ be a $[q+1,k,d=q+2-k]_q$ MDS code of length $n=q+1$ and minimum distance $d\ge3$. For $\CC$, the overall number $\A_{w}^{\Sigma}(\V^{(1)})$ of weight $w$ vectors in all weight $1$ cosets is as follows
\begin{align}\label{eq4_Aw1sum_q+1}
 &\A_{w}^{\Sigma}(\V^{(1)})=\binom{q+1}{w}(q-1)\left[q^{w+2-d}-\sum_{i=0}^{w-d}(-1)^{i}\left(\binom{w}{i+1}-\binom{w}{i}\right)q^{w+1-d-i}\right.\\
 &\left.-(-1)^{w-d}\left(\binom{w}{d-1}-w\binom{w-2}{d-3}\right)\right],~d-1\le w\le q+1.\notag
\end{align}

\textbf{\emph{(iii)}} Let the code $\CC$ be a $[q+1,q-3,5]_q$ MDS code of length $n=q+1$ and minimum distance  $d=5$. For $\CC$, the overall number $\A_{w}^{\Sigma}(\V^{(1)})$ of weight $w$ vectors in all weight $1$ cosets is as follows
\begin{align}\label{eq4_Aw1sum_q+1_5}
\A_{w}^{\Sigma}(\V^{(1)})=(q^2-1)\left[A_w(\CC)-\Phi_w^{(1)}\right],~4\le w\le q+1.
\end{align}
\end{theorem}
\begin{proof}
\textbf{(i)} By the definition of $\A_w^{\Sigma}(\V^{\le 1})$, see Notation \ref{notation_coset}, for the code $\CC$ of Lemma \ref{lem4_AwSigma<=1MDS}, we have
  \begin{align}\label{eq4_<=&1}
 \A_{w}^{\Sigma}(\V^{(1)})=\A_{w}^{\Sigma}(\V^{\le1})-A_w(\CC).
  \end{align}
We subtract  \eqref{eq2_wd_MDS} from \eqref{eq4_AwSigma<=1MDS} that gives
\begin{align*}
&\A_{w}^{\Sigma}(\V^{(1)})=\binom{n}{w}(q-1)\left[-(-1)^{w-d}\binom{w}{d-1}(n-d+1)\right.\displaybreak[3]\\
&\left.+\sum_{j=0}^{w-d}(-1)^j\binom{w}{j}\left(q^{w-d+1-j}n-w+j\right)\right]\displaybreak[3]\\
&=\binom{n}{w}(q-1)\left[n\sum_{j=0}^{w-d+1}(-1)^j\binom{w}{j}q^{w-d+1-j}-\sum_{j=0}^{w-d+1}(-1)^j\binom{w}{j}(w-j)\right].
 \end{align*}
Here some simple transformations are missed out. Now, for the 2-nd sum $\sum_{j=0}^{w-d+1}\ldots$, we use Lemma \ref{lem4_Riord} and obtain \eqref{eq4_Aw1sum}.

To form \eqref{eq4_Aw1sum_2} from \eqref{eq4_Aw1sum}, we change $w\binom{n}{w}$ by $n\binom{n-1}{w-1}$, see \eqref{eq2_Riordan_ident}.
To obtain \eqref{eq4_Aw1sum_3}  from~\eqref{eq4_Aw1sum_2}, we apply Lemma \ref{lem4_-1}.
For \eqref{eq4_Aw1sum_4}, we use \eqref{eq2_wd_MDS}.
Finally, \eqref{eq4_Aw1sum_5} is \eqref{eq4_Aw1sum_4} in detail.

\textbf{(ii)}  We substitute $n=q+1$ to  \eqref{eq4_Aw1sum} that implies \eqref{eq4_Aw1sum_q+1} after simple transformations.

\textbf{(iii)}  We substitute $n=q+1$ and $d=5$ to  \eqref{eq4_Aw1sum_5} that gives \eqref{eq4_Aw1sum_q+1_5}.
\end{proof}

For $\A_w^{\Sigma}(\V^{\le1})$, we give a formula alternative to \eqref{eq4_AwSigma<=1MDS}.
\begin{corollary}
Let $V_n(1)$ be as in \eqref{eq2_sphere}. Let $\CC$ be an $[n,k,d=n-k+1]_q$ MDS code of minimum distance $d\ge3$.  Then for  $\CC$,  the overall number $\A_{w}^{\Sigma}(\V^{\le1})$ of weight $w$ vectors in all cosets of weight $\le1$ is as follows:
\begin{align}\label{eq4_AwSigma<=1MDS_new}
&\A_w^{\Sigma}(\V^{\le1})=A_w(\CC)\cdot V_n(1)-(-1)^{w-d}n(q-1)\sum_{j=0}^1(-1)^j\binom{n-j}{w-j}\binom{w-j-1}{d-j-2}.
 \end{align}
\end{corollary}
\begin{proof}
  We use \eqref{eq4_<=&1} and \eqref{eq4_Aw1sum_5}.
\end{proof}

\section{The integral weight spectrum of the weight 2 cosets of MDS codes with minimum distance $d\ge5$} \label{sec_wd2}

As well as in Lemma \ref{lem4_AwSigma<=1MDS}, we use the results of \cite{CheungIEEE1989} with some transformations.
\begin{lemma}
\label{lem5_AwSigma<=2MDS}
\emph{\cite[Eqs.\,(11)--(13)]{CheungIEEE1989}}
Let $d-2\le w\le n$.  Let $V_n(t)$ be as in \eqref{eq2_sphere}. For an $[n,k,d=n-k+1]_q$ MDS code $\CC$ of minimum distance $d\ge5$,  the overall number $\A_{w}^{\Sigma}(\V^{\le2})$ of weight $w$ vectors in all cosets of weight $\le2$ is as follows:
\begin{align}\label{eq5_AwSigma<=2MDS}
&\A_w^{\Sigma}(\V^{\le2})=\binom{n}{w}\left[\sum_{j=0}^{w-d}(-1)^j\binom{w}{j}
\left[q^{w-d+1-j}\cdot V_n(2)-V_{w-j}(2)\right]\right.\\
&\left.-(-1)^{w-d}\frac{(n-d+1)(q-1)}{2}\left(\binom{w}{d-1}[2+(q-1)(n+d-2)]
-\binom{w}{d-2}(n-d+2)\right)\right].\notag
 \end{align}
 \end{lemma}
\begin{proof}
  In the relations for $D_u$ of \cite{CheungIEEE1989} cited by \eqref{eq2_Cheung}--\eqref{eq2_Cheung_3}, we put $t=2$ that gives, in \eqref{eq2_Cheung_3}, $j=u-d+1$ and $j=u-d+2$, whence $w=1,2$ and $w=2$, respectively.  Then we do simple transformations.  Finally, we change $u$ by $w$ to save the notations of this paper.
\end{proof}

For an $[n,k,d]_q$ code $\CC$, we denote
\begin{align}
&\Delta_w(\CC)=(-1)^{w-d}\binom{n}{w}\binom{w}{d-2}\binom{n-d+2}{2}(q-1);\displaybreak[3]\label{eq5_Delta}\\
&\Delta_w^\star(\CC)=\frac{\Delta_w(\CC)}{\binom{n}{2}(q-1)^2}.\notag
\end{align}
\begin{lemma} \label{lem5_DeltaStar}
The following holds:
\begin{align}\label{eq5_DeltaStar}
  \Delta_w^\star(\CC)=(-1)^{w-d}\binom{n-d+2}{n-w}\binom{n-2}{d-2}\frac{1}{q-1}.
\end{align}
\end{lemma}
\begin{proof}
  By \eqref{eq2_Riordan_ident}, we have
  \begin{align*}
&\binom{n}{w}\binom{w}{d-2}=\binom{n}{d-2}\binom{n-d+2}{w-d-2}=\binom{n}{d-2}\binom{n-d+2}{n-w},\displaybreak[3]\\
&\binom{n}{d-2}\binom{n-d+2}{2}=\binom{n}{d}\binom{d}{d-2}=\binom{n}{d}\binom{d}{2}=\binom{n}{2}\binom{n-2}{d-2}.
  \end{align*}
\end{proof}
\begin{theorem}\label{th5_Aw2Sigma}\emph{\textbf{(integral weight spectrum 2)}}

 Let $d-2\le w\le n$.
Let $\CC$ be an $[n,k,d=n-k+1]_q$ MDS code of minimum distance $d\ge5$.  Let $\Omega_w^{(j)}(\CC)$  and $\Phi_w^{(j)}$ be as in \eqref{eq4_Omega} and \eqref{eq4_Phi}.

\textbf{\emph{(i)}} For the code~$\CC$, the overall number $\A_{w}^{\Sigma}(\V^{(2)})$ of weight $w$ vectors in all weight $2$ cosets is as follows:
\begin{align}\label{eq5_AwSigma2}
&\A_{w}^{\Sigma}(\V^{(2)})=\binom{n}{w}(q-1)^2\left[\binom{n}{2}\sum_{j=0}^{w+1-d}(-1)^j\binom{w}{j}q^{w+1-d-j}+(-1)^{w-d}\binom{w}{2}\binom{w-3}{d-4}\right]\displaybreak[3]\\
&\phantom{\A_{w}^{\Sigma}(\V^{(2)})=}+\Delta_w(\CC).\displaybreak[3]\notag\\
&=\binom{n}{2}(q-1)^2\left[\binom{n}{w}\sum_{j=0}^{w+1-d}(-1)^j\binom{w}{j}q^{w+1-d-j}+\Omega_w^{(2)}(\CC)\right]+\Delta_w(\CC).\label{eq5_AwSigma2_2}\displaybreak[3]\\
&=\binom{n}{2}(q-1)^2\left[\binom{n}{w}\sum_{j=0}^{w-d}(-1)^j\binom{w}{j}\left(q^{w+1-d-j}-1\right)-\Omega_w^{(0)}(\CC)+\Omega_w^{(2)}(\CC)\right]
+\Delta_w(\CC)\label{eq5_AwSigma2_3}\displaybreak[3]\\
&=\binom{n}{2}(q-1)^2\left[A_w(\CC)-\Omega_w^{(0)}(\CC)+\Omega_w^{(2)}(\CC)\right]
+\binom{n}{2}(q-1)^2\Delta_w^\star(\CC)\label{eq5_AwSigma2_4}\displaybreak[3]\\
&=\binom{n}{2}(q-1)^2\left[A_w(\CC)-(-1)^{w-d}\left(\binom{n}{w}\binom{w-1}{d-2}-\binom{n-2}{w-2}\binom{w-3}{d-4}\right)\right]\label{eq5_AwSigma2_5}\\
&+(-1)^{w-d}\binom{n}{2}(q-1)\binom{n-d+2}{n-w}\binom{n-2}{d-2}.\notag
\end{align}

\textbf{\emph{(ii)}} Let the code $\CC$ be a $[q+1,q-3,5]_q$ MDS code of length $n=q+1$ and minimum distance  $d=5$. For $\CC$, the overall number $\A_{w}^{\Sigma}(\V^{(1)})$ of weight $w$ vectors in all weight $1$ cosets is as follows
\begin{align}\label{eq5_AwSigma2_q+1_5}
&\A_{w}^{\Sigma}(\V^{(2)})=\binom{q+1}{2}(q-1)^2\left[A_w(\CC)-\Phi_w^{(2)}+(-1)^{w-5}\,\frac{1}{3}\binom{q-2}{w-3}\binom{q-2}{2}\right],\displaybreak[3]\\
&\phantom{\A_{w}^{\Sigma}(\V^{(2)})=}3\le w\le q+1. \notag
\end{align}
\end{theorem}
\begin{proof}
\textbf{(i)}
By the definition of $\A_w^{\Sigma}(\V^{\le 2})$, see Notation \ref{notation_coset}, for the code $\CC$ of Lemma \ref{lem5_AwSigma<=2MDS}, we have
  \begin{align}\label{eq5_extract}
 \A_{w}^{\Sigma}(\V^{(2)})=\A_{w}^{\Sigma}(\V^{\le2})-\A_{w}^{\Sigma}(\V^{\le1}).
  \end{align}
We subtract  \eqref{eq4_AwSigma<=1MDS} from \eqref{eq5_AwSigma<=2MDS} that gives
\begin{align*}
&\A_{w}^{\Sigma}(\V^{(2)})=\binom{n}{w}\left[\sum_{j=0}^{w-d}(-1)^j\binom{w}{j}\left(q^{w+1-d-j}\binom{n}{2}(q-1)^2-\binom{w-j}{2}(q-1)^2\right)\right.\displaybreak[3]\\
&\left.+(-1)^{w+1-d}\binom{w}{d-1}\frac{1}{2}(n-d+1)(q-1)^2(n+d-2)\right]+\Delta_w(\CC)\displaybreak[3]\\
&=\binom{n}{w}(q-1)^2\left[\binom{n}{2}\sum_{j=0}^{w-d}(-1)^j\binom{w}{j}q^{w+1-d-j}-\sum_{j=0}^{w-d}(-1)^j\binom{w}{j}\binom{w-j}{2}\right.\displaybreak[3]\\
&\left.-(-1)^{w-d}\binom{w}{d-1}\left(\frac{1}{2}(n-d+1)(n+d-2)+\binom{n}{2}-\binom{n}{2}\right)\right]+\Delta_w(\CC).
\end{align*}
Applying  Lemma \ref{lem4_Riord} to the 2-nd sum $\sum_{j=0}^{w-d}\ldots$, after simple transformations we obtain
\begin{align*}
&\A_{w}^{\Sigma}(\V^{(2)})=\binom{n}{w}(q-1)^2\left[\binom{n}{2}\sum_{j=0}^{w+1-d}(-1)^j\binom{w}{j}q^{w+1-d-j}-(-1)^{w-d}\binom{w}{2}\binom{w-3}{w-d}\right.\displaybreak[3]\\
&\left.+(-1)^{w-d}\binom{w}{d-1}\binom{d-1}{2}\right]+\Delta_w(\CC).
\end{align*}
Due to \eqref{eq2_Riordan_ident} and \eqref{eq2_Riordan_ident0}, we have
$$\binom{w}{d-1}\binom{d-1}{2}=\binom{w}{2}\binom{w-2}{d-3}=\binom{w}{2}\left[\binom{w-3}{d-4}+\binom{w-3}{d-3}\right].$$ Also, $\binom{w-3}{w-d}=\binom{w-3}{d-3}$. Now we can obtain \eqref{eq5_AwSigma2}. Moreover, by \eqref{eq2_Riordan_ident}, we have $$\binom{n}{w}\binom{w}{2}=\binom{n}{2}\binom{n-2}{w-2}$$ that gives \eqref{eq5_AwSigma2_2}.

To obtain \eqref{eq5_AwSigma2_3} from \eqref{eq5_AwSigma2_2}, we apply Lemma \ref{lem4_-1}.
For \eqref{eq5_AwSigma2_4}, we use \eqref{eq2_wd_MDS}.
Finally, \eqref{eq5_AwSigma2_5} is \eqref{eq5_AwSigma2_4} in detail.

\textbf{(ii)} We substitute $n=q+1$ and $d=5$ to  \eqref{eq5_AwSigma2_5} that gives \eqref{eq5_AwSigma2_q+1_5}.
\end{proof}

\section{The integral weight spectrum of the weight 3 cosets of MDS codes with minimum distance $d=5$ and covering radius $R=3$} \label{sec_wd3}

\begin{theorem}\label{th6_integr_3} \emph{\textbf{(integral weight spectrum 3)}}

 Let $d-2\le w\le n$.  Let $\CC$ be an $[n,n-4,5]_q3$ MDS code of minimum distance $d=5$ and covering radius $R=3$. Let $V_n(t)$, $\Phi_w^{(j)}$, $\A_w^{\Sigma}(\V^{\le2})$, and $\Delta_w(\CC)$ be as in \eqref{eq2_sphere}, \eqref{eq4_Phi}, \eqref{eq5_AwSigma<=2MDS}, and \eqref{eq5_Delta}, respectively. Let
 $\A_{w}^{\Sigma}(\V^{(1)})$ and $\A_{w}^{\Sigma}(\V^{(2)})$ be as in Theorems \emph{\ref{th4_Aw1sum}} and \emph{\ref{th5_Aw2Sigma}}, respectively.

\textbf{\emph{(i)}} For  the code~$\CC$,  the overall number $\A_{w}^{\Sigma}(\V^{3})$ of weight $w$ vectors in all cosets of weight~$3$ is as follows:
 \begin{align}\label{eq6_wd3_gen}
 &\A_w^{\Sigma}(\V^{(3)})=\binom{n}{w}(q-1)^w-\A_w^{\Sigma}(\V^{\le2})\displaybreak[3]\\
 &=\binom{n}{w}(q-1)^w-\left[A_w(\CC)+\A_w^{\Sigma}(\V^{(1)})+\A_w^{\Sigma}(\V^{(2)})\right]\displaybreak[3]\label{eq6_wd3_a}\\
 &=\binom{n}{w}(q-1)^w-\left[\binom{n}{w}\sum_{j=0}^{w-5}(-1)^j\binom{w}{j}
\left[q^{w-4-j}\cdot V_n(2)-V_{w-j}(2)\right]\right.\label{eq6_wd3_b}\\
&\left.-(-1)^{w-5}\frac{(n-4)(q-1)}{2}\left(\binom{w}{4}[2+(q-1)(n+3)]
-\binom{w}{3}(n-3)\right)\right].\notag
   \end{align}

\textbf{\emph{(ii)}} Let the code $\CC$ be a $[q+1,q-3,5]_q3$ MDS code of length $n=q+1$, minimum distance  $d=5$, and covering radius $R=3$. For $\CC$, the overall number $\A_{w}^{\Sigma}(\V^{(3)})$ of weight $w$ vectors in all weight $3$ cosets is as follows
\begin{align}\label{eq6_wd3_q+1}
&\A_w^{\Sigma}(\V^{(3)})=\binom{q+1}{w}(q-1)^w-\left[\binom{q+1}{w}\sum_{j=0}^{w-5}(-1)^j\binom{w}{j}
\left[q^{w-4-j}\cdot V_{q+1}(2)-V_{w-j}(2)\right]\right.\\
&\left.-(-1)^{w-5}\frac{(q-3)(q-1)}{2}\left(\binom{w}{4}(q^2+3q-2)
-\binom{w}{3}(q-2)\right)\right]\notag\\
 &=\binom{q+1}{w}(q-1)^w-\left[V_{q+1}(2)A_w(\CC)-(q^2-1)\Phi_w^{(1)}-\binom{q+1}{2}(q-1)^2\Phi_w^{(2)}-\Delta_w(\CC)\right].\label{eq6_wd3_q+1_b}
\end{align}
\end{theorem}
\begin{proof}
\textbf{(i)} Due to covering radius $3$,  in $\CC$ there are not cosets of weight $>3$; therefore for $\CC$ we have \eqref{eq6_wd3_gen}
 where $\binom{n}{w}(q-1)^w$ is the total number of weight $w$ vectors in $\F_q^{n}$.

 The relation \eqref{eq6_wd3_a} follows from \eqref{eq6_wd3_gen}, \eqref{eq4_<=&1}, and \eqref{eq5_extract}.

 To form \eqref{eq6_wd3_b}, we substitute \eqref{eq5_AwSigma<=2MDS} to \eqref{eq6_wd3_gen} with $d=5$.

\textbf{(ii)} We substitute $n=q+1$ to \eqref{eq6_wd3_b} and obtain \eqref{eq6_wd3_q+1}.

To obtain \eqref{eq6_wd3_q+1_b} from \eqref{eq6_wd3_a}, we use \eqref{eq4_Aw1sum_q+1_5}, \eqref{eq5_AwSigma2_q+1_5}, \eqref{eq5_Delta}, and \eqref{eq5_DeltaStar} with $n=q+1$, $d=5$.
\end{proof}

\end{document}